\documentclass[12pt]{article}
\usepackage{amssymb}
\usepackage{amsmath,graphicx,amsfonts,epsfig,eucal}
\usepackage{amsfonts}
\usepackage{color}
\usepackage{amsthm}

\usepackage[utf8]{inputenc}
\usepackage[english]{babel}

\newtheorem{prop}{Proposition}
\newtheorem{defn}{Definition}
\newtheorem{rem}{Remark}

\begin{document}
\title{On the stability of the steady-state of a general model of endogenous growth with two $CES$ production functions}
\author{Constantin Chilarescu\footnote{e-mail: constantin.chilarescu@univ-lille.fr}}
\date{}
\maketitle
\centerline{\it University of Lille, France}
\maketitle
\noindent {\small {\bf Abstract} The main aim of this paper is to study the steady-state properties of a general Bond-type endogenous growth model, considering that both sectors are modeled by two distinct $CES$ production functions. We prove here that in this case, we cannot claim the saddle-path stability.}

\noindent {\small {\bf JEL\ Classifications:} C61, C62, E21, O41.}
\date{}
\maketitle
\section{Introduction}
In a reference paper on endogenous growth, Bond et al. $(1996)$ proposed a general two-sector model of endogenous growth, consisting of a goods sector and an education sector, in which each sector produces the
respective type of capital (i.e., physical or human) under conditions of constant returns to scale. In their paper, the authors claim that their approach differs from that of other work in this area, which has generally approached the problem by using simulation with specific functional forms, as for example Rebelo $(1991)$ or Benhabib and Perli $(1994)$. They also claim that they prove that the model is characterized by the existence and uniqueness of the balanced growth equilibrium and show that the system will be saddle-path stable. In this paper we will prove that this last statement of the paper of Bond et al. is not true for any choice of the two production functions. To do this we will consider the case of a model with two distinct $CES$ production functions.

The paper is organized as follows. In the next section of this paper presents the model with two distinct $CES$ production functions. In the third section we prove the existence and uniqueness of the steady-state equilibrium and in the fourth section we present our main result. The final section gives some conclusions.
\section{The general two-sector model}
The two sectors considered in this paper are the good sector that produces consumable and gross investment in physical capital, and the education sector that produces human capital, both of them under conditions
of constant returns to scale. Without loss of generality, we suppose that the economy is populated by a large and constant number of identical agents, normalized to one, who seek to maximize their utility function.
The good sector combines a fraction $v$ of the stock of physical capital with a fraction $u$ of the stock of human capital.
The output of the educational sector increases the supply of effective labor units by adding to the existing stock of human capital, that will then be used in both sectors. The educational sector combines the remaining fractions of stocks.
Our model is characterized by the following optimization problem.
\begin {defn}
The set of paths $\left\{k, h, c, u, v\right\}$ is called an optimal
solution if it solves the following optimization problem:
\begin {equation}\label{eqof}
V_0 = \max
\int\limits_0^{\infty}\frac{c^{1-\varepsilon}-1}{1-\varepsilon}e^{-\rho t}dt,
\end {equation}
\noindent subject to
\begin {equation}\label{eqres}
\left\{
 \begin{array}{lll}
\dot{k} = A_1\left[\alpha_1\,\left(kv\right)^{\psi_1}+\left(1-\alpha_1
\right)\left(hu\right)^{\psi_1}\right]^{\frac{1}{\psi_1}} - c-\delta_{k} k,\\
\dot{h} = A_2\left[\alpha_2\,\left[k\left(1-v\right)\right]^{
\psi_2}+\left(1-\alpha_2\right)\left[h\left(1-u\right)\right]
^{\psi_2}\right]^{\frac{1}{\psi_2}}-\delta_{h} h,\\
k_0 = k(0),\;h_0 = h(0).
  \end{array}
  \right.
\end {equation}
\end {defn}
\noindent In this model $h$ represents the human capital
or the skill level and $k$ is capital per labor unit. $\alpha_1 \in (0,1)$ and $\alpha_2 \in (0, 1)$ are distribution parameters,
$\psi_{i} = \frac{\sigma_{i}-1}{\sigma_{i}},\;\psi_{i} < 1,\; i = 1, 2$ are the substitution parameters.
Since the two elasticities of substitution of the two sectors are different $\sigma_1 \neq \sigma_2$, the same hypothesis is to be considered for the two substitution parameters $\psi_1 \neq \psi_2$.
$\rho$ is the rate of time preference, $\delta_{k}$ and $\delta_{h}$ are the depreciation rates,
$A_1$ and $A_2$ are efficiency parameters, and $c \geq 0$ is
the real per-capita consumption. $\varepsilon^{-1}$ represents the
constant elasticity of intertemporal substitution and $k_0$ and $h_0$ are given values.
In order to simplify the writing procedure, we use the following notations:
$$\bar{P}_1=\alpha_1\left(kv\right)^{\psi_1}+\left(1-\alpha_1
\right)\left(hu\right)^{\psi_1},\bar{P}_2=\alpha_2\left[k\left(1-v\right)\right]^{
\psi_2}+\left(1-\alpha_2\right)\left[h\left(1-u\right)\right]^{\psi_2}$$
To solve the problem \eqref{eqof} subject to  \eqref{eqres}, we define the
Hamiltonian function:
\begin{equation}\label{eqham}
H = \frac{c^{1-\varepsilon}-1}{1-\varepsilon} + \left\{A_1\left[\bar{P}_1\right]^{\frac{1}{\psi_1}} - c-\delta_{k} k\right\}\lambda + \left\{A_2\left[\bar{P}_2\right]^{\frac{1}{\psi_2}}-\delta_{h} h\right\}\mu.
\end{equation}
\noindent The boundary conditions include initial conditions $\bar{h},
\;\bar{k}$ and transversality condition:
$$\lim\limits_{t\rightarrow\infty} e^{-\rho t}\lambda(t)
k(t) = 0 \;\;\mbox{and}\;\; \lim\limits_{t\rightarrow\infty}
e^{-\rho t}\mu(t) h(t) = 0.$$
In our model there are three control variables, $c$, $v$ and $u$, and two
state variables, $k$ and $h$. Differentiating the Hamiltonian function
with respect to $u$ and then with respect to $v$, we obtain the
following relations:
\begin{equation}\label{eqc}
c^{-\varepsilon} = \lambda,
\end{equation}
\begin{equation}\label{eqderu}
\frac{A_2\left[\bar{P}_2\right]^{\frac{1}{\psi_2}-1}\left(1-\alpha_2 \right)
\left[h\left(1-u\right)\right]^{\psi_2}\mu}{1-u}=\frac {A_1\left[\bar{P}_1\right]^{\frac{1}{\psi_1}-1}\left(1-\alpha_1\right)\left(hu\right)^{\psi_1}\lambda}{u}
\end{equation}
\begin{equation}\label{eqderv}
\frac {A_2\left[\bar{P}_2\right]^{\frac{1}{\psi_2}-1}\alpha_2\left[k\left(1-v\right)\right]^{\psi_2}\mu}{1-v
}=\frac{A_1\left[\bar{P}_1\right]^{\frac{1}{\psi_1}-1}\alpha_1\,\left(kv\right)^{\psi_1}\lambda}{v}.
\end{equation}
From equation \eqref{eqderv}, we get:
\begin{equation}\label{eqlm}
\frac{\mu}{\lambda} = \frac{A_1\alpha_1\theta^{\psi_1-\psi_2}}{A_2\alpha_2}\frac{P_1^{\frac{1}{\psi_1}-1}}{P_2^{\frac{1}{\psi_2}-1}},
\end{equation}
where $z > 0$, $w > 0$
\begin{equation}\label{eqsr}
 z = \frac{k}{h},\;z\frac{v}{u}=\theta w,\;w=
\left[\frac{v(1-u)}{u(1-v)}\right]^{\frac{1-\psi_2}{\psi_1-\psi_2}},\;
\theta=\left[\frac{\alpha_2(1-\alpha_1)}{\alpha_1(1-\alpha_2)}\right]^{\frac{1}{\psi_1-\psi_2}},
\end{equation}
and,
\begin{equation}\label{eqp1p2}
P_1=\alpha_1 \theta^{\psi_1} w^{\psi_1}+1-\alpha_1,\; P_2=\alpha_2\theta^{\psi_2}
w^{\frac{\psi_2(1-\psi_1)}{1-\psi_2}}+1-\alpha_2.
\end{equation}
Differentiating the Hamiltonian function with respect to $k$ and then with respect to $h$, we obtain
the differential equations describing the trajectories of the dual variables.
\begin{equation}\label{eqsysb1}
\frac{\dot{\lambda}}{\lambda}=\rho+\delta_{k}-\alpha_1A_1\left(\theta w\right)^{\psi_1-1}P_1^{\frac{1}{\psi_1}-1},\;
\frac{\dot{\mu}}{\mu}=\rho+\delta_{h}-\left(1-\alpha_2\right)A_2 P_2^{\frac{1}{\psi_2}-1}.
\end{equation}
Differentiating with respect to time the equations \eqref{eqlm} and \eqref{eqsr} we obtain the following differential equations:
\begin{equation}\label{eqdifuv1}
\frac{G_1}{1-u}\frac{\dot{u}}{u}-\frac{G_2}{1-v}\frac{\dot{v}}{v}=
(\psi_1-\psi_2)\left(D+\frac{c}{k}\right),
\end{equation}
\begin{equation}\label{eqdifuv2}
\frac{\dot{u}}{u}-\frac{\dot{v}}{v}=-\left(D+\frac{c}{k}\right)-\frac{QP}{\left(1-\psi_1\right)T},
\end{equation}
where
$$\frac{\dot{k}}{k}-\frac{\dot{h}}{h}=-\left(D+\frac{c}{k}\right),
\;\frac{\dot{\mu}}{\mu}-\frac{\dot{\lambda}}{\lambda}=P,$$
$$D=A_2(1-u)P_2^{\frac{1}{\psi_2}}
-A_1v\left(\theta w\right)^{-1}P_1^{\frac{1}{\psi_1}}+\delta_k-\delta_h,$$
$$P=\alpha_1A_1\left(\theta w\right)^{\psi_1-1}P_1^{\frac{1}{\psi_1}-1}
-\left(1-\alpha_2\right)A_2P_2^{\frac{1}{\psi_2}-1}
-\left(\delta_k-\delta_h\right),$$
$$T=\alpha_1\left(1-\alpha_2\right)\left(\theta w\right)^{\psi_1}\frac{w^\frac{\psi_1-\psi_2}{1-\psi_2}-1}{w^\frac{\psi_1-\psi_2}{1-\psi_2}},
\;Q=P_1P_2,\;R=\left(1-\psi_1\right)\left(1-\psi_2\right)T.$$
$$G_1=(\psi_1-\psi_2)u+1-\psi_1,\;G_2=(\psi_1-\psi_2)v+1-\psi_1,$$
Now, we can close the system and write down the final form of the five differential equations,
\begin{equation}\label{eqsysb}
\left\{
  \begin{array}{llllll}
\frac{\dot{k}}{k} = A_1v\left(\theta w\right)^{-1}P_1^{\frac{1}{\psi_1}} - \frac{c}{k}-\delta_{k},\\\\
\frac{\dot{h}}{h} =A_2P_2^{\frac{1}{\psi_2}}(1-u) -\delta_{h},\\\\
\frac{\dot{c}}{c} = -\frac{\rho+\delta_{k}}{\varepsilon}+\frac{\alpha_1A_1\left(\theta w\right)^{\psi_1-1}}{\varepsilon}P_1^{\frac{1}{\psi_1}-1},\\\\
\frac{\dot{u}}{u}=\left[D+\frac{c}{k}+\frac{Q G_2P}{R}\right]\frac{1-u}{u-v},\\\\
\frac{\dot{v}}{v}=\left[D+\frac{c}{k}+\frac{Q G_1P}{R}\right]\frac{1-v}{u-v}.\\
\end{array}
  \right.
\end{equation}
The last two equations of the above system result immediately from the equations \eqref{eqdifuv1} and  \eqref{eqdifuv2}. As we can observe from the definition of the variable $w$, the function $R$ is always different from zero, except when $w = 1$ that is, in other words, when $u = v$. Without loss of generality, everywhere in this paper, we assume that $u \neq v$.
\section{The balanced growth path}
The system described above reaches the balanced growth path if
there exists $t_* > 0$, such that for all $t \geq t_*$, the growth rates of the control variables
$u$ and $v$ are equal to zero.
To simplify the notation we denote by $r_x$ the growth
rate of variable $x$, by $x^*$ its value along the balanced growth
path $(t \geq t_*)$, and by $x_*$ its value for $t = t_*$. The following
proposition gives our first main result.
\begin{prop}\label{prop1}
Let $$\left(\varepsilon-1\right)\left[\delta_h-A_2\left(1-\alpha_2\right)^{\frac{1}{\psi_2}}\right]<\rho$$ $$<
\min\left\{A_1\alpha_1^{\frac{1}{\psi_1}}-\delta_k,A_2(1-\alpha_2)
^{\frac{1}{\psi_2}}+(\varepsilon-1)\delta_h\right\}.$$ If for all $t \geq t_* $,\; $r_u = r_v = 0$, then the above system reaches the unique balanced growth path and the following statements are valid
\begin{enumerate}
\item[i.] $r_{k^*} = r_{h^*} = r_{c^*}  = r_{y^*_1} = r_{y^*_2} = r_*$, where $y_1$ and $y_2$ represent the two $CES$ production functions and
\begin{equation}\label{eqsolr}
r_*=\frac{1}{\varepsilon}\left[\alpha_1A_1\left(\theta w\right)^{\psi_1-1}P_1^{\frac{1}{\psi_1}-1}
-\rho-\delta_{k}\right]
\end{equation}
\item [ii.] $w_{*}$ is the unique solution of the equation
\begin{equation}\label{eqsolz}
P(w) =\alpha_1A_1\left(\theta w\right)^{\psi_1-1}P_1^{\frac{1}{\psi_1}-1}-\left(1-\alpha_2\right)A_2P_2^{\frac{1}{\psi_2}-1}
-\left(\delta_k-\delta_h\right)=0,
\end{equation}
\item[iii.] $u_* \in [0, 1]$, $v_* \in [0, 1]$, $\tau=w^{\frac{\psi_1-\psi_2}{1-\psi_2}}$ and
\begin{equation}\label{eqsolu}
u_* = 1-\frac{r_*+\delta_h}{A_2P^{\frac{1}{\psi_2}}_2},\;v_* = \frac{\tau u_*}{1+\left(\tau-1\right)u_*},
\end{equation}
\item[iv.]
\begin{equation}\label{eqsolc/k}
q_*=\frac{c_*}{k_*} = \frac{1}{\varepsilon} \left[A_1\left(\theta w\right)^{-1}P_{\varepsilon}P_1 ^{\frac{1}{\psi_1}-1}+\rho-\delta_k\left(\varepsilon-1\right)\right],
\end{equation}
where $$P_{\varepsilon}=\alpha_1(\varepsilon v_*-1)\left(\theta w\right)^{\psi_1}+\varepsilon(1-\alpha_1)v_*.$$
\end{enumerate}
\end{prop}
\noindent {\bf Proof of Proposition 1.}
If $r_u=r_v =0$ for all $t \geq t_*$ and since $P_1P_2 \neq 0$ for all $t \geq t_*$, from equations \eqref{eqdifuv1} and \eqref{eqdifuv2}, it immediately follows that $r_{k^*} = r_{h^*}$ and $r_{\mu^*}= r_{\lambda^*}$.
Of course, because we have $r_{k_*} = r_{h_*}$, we also must have $r_{y_{1*}}=r_{y_{2*}}=r_{k_*}$. This statement is obtained simply by log-differentiating the two production functions. From the two equations of \eqref{eqsysb1} we obtain \eqref{eqsolz} and $w_*$ will be the unique solution of this equation. (It is just a simply exercise to prove that this equation has a unique solution. Indeed, $\lim\limits_{w\rightarrow 0}P(w) = +\infty$, $\lim\limits_{w\rightarrow +\infty}P(w) = -\infty$ and $P^{\prime} < 0$ for all $w > 0$. Consequently, the function $P$ is a strictly decreasing function on $(0, \infty)$ and therefore there exists a unique solution of the equation $P(w) = 0.$) Substituting now $w_*$ into the third equation of the system \eqref{eqsysb} we get the common growth rate of the economy given by \eqref{eqsolr}. Let us denote $$S_1 = A_1\alpha_1\left(\theta w\right)^{\psi_1-1}P_1^{\frac{1}{\psi_1}-1}-\rho-\delta_k,$$ and we have:
$\lim\limits_{w\rightarrow 0}S_1= +\infty$ and $\lim\limits_{w\rightarrow +\infty}S_1= A_1\alpha_1^{\frac{1}{\psi_1}}-\rho-\delta_k > 0.$ Also the first derivative of the function $S_1$ $wrt$ $w$ is strictly negative and thus, the function $S_1$ is a strictly decreasing positive function and therefore, the common growth rate will be positive. Substituting $r_*$ from equation \eqref{eqsolr} into the second equation of  the system \eqref{eqsysb} we obtain \eqref{eqsolu}. We can now use the equation \eqref{eqsolz} in order to rewrite the equation \eqref{eqsolu} to obtain
$$u_* = 1-\frac{\left(1-\alpha_2\right)A_2P_2^{\frac{1}{\psi_2}-1}-\rho+\left(\varepsilon-1\right)\delta_h}{\varepsilon A_2 P_2^{\frac{1}{\psi_2}}}.$$
If we denote by $S_2 = \left(1-\alpha_2\right)A_2P_2^{\frac{1}{\psi_2}-1}-\rho+\left(\varepsilon-1\right)\delta_h$, then we observe that $S_2(0) = A_2 \left(1-\alpha_2\right)^{\frac{1}{\psi_2}}-\rho+(\varepsilon-1)\delta_h>0$. Also, $\lim\limits_{w\rightarrow +\infty}S_2(w)=+\infty$. The derivative of the function $S_2$ $wrt$ $w$ is obviously positive and therefore the function $S_2$ is a positive increasing function.
Consequently, to prove that $u_* \in (0, 1)$, it is enough to prove that the function $S_3 = \varepsilon A_2 P_2^{\frac{1}{\psi_2}}-S_2$ is a positive function for all $w > 0$. Indeed, we have $S_3(0) = (\varepsilon-1)A_2(1-\alpha_2)^{\frac{1}{\psi_2}}+\rho-(\varepsilon-1)\delta_h > 0$ and $\lim\limits_{w\rightarrow +\infty}S_3(w)=+\infty$. The derivative of the function $S_3$ $wrt$ $w$ is positive for all $w > 0$ and therefore the function $S_3$ is a positive increasing function and thus we proved that $u_* \in (0, 1)$. We can now substitute the result for $u_*$ into the first equation of the system \eqref{eqsysb} to obtain the steady state value for the variable $\frac{c}{k}$ given in equation \eqref{eqsolc/k}.
Under the hypotheses $u_*\in (0, 1)$ it immediately follows that $v_*\in (0, 1)$ and thus the proof is completed.$\square$

We have now to verify if the steady-state found above, satisfies the transversality conditions. Equivalently, the two conditions are true if the following limits are both negative, that is
$l_1=\lim\limits_{t\rightarrow\infty}\left(\frac{\dot{k}}{k}+\frac{\dot{\lambda}}{\lambda}-\rho\right)<0$ and $l_2=\lim\limits_{t\rightarrow\infty}\left(\frac{\dot{h}}{h}+\frac{\dot{\mu}}{\mu}-\rho\right)<0.$
Substituting the corresponding elements from the system \eqref{eqsysb} and passing to the limit for $t\rightarrow\infty$ we obtain for both limits the following result:
$l_1= l_2 = -\left[\rho+\left(\varepsilon-1\right)r_*\right].$
This limit is obviously negative, for all $\varepsilon > 1$, that is if the inverse of the elasticity of intertemporal substitution is greater than one. Without loss of generality, everywhere in this paper we assume that this inequality holds and therefore, since both transversality conditions are true, we may conclude that the solution obtained for the steady-state, is the unique optimal solution.
\section{Stability analysis}
As we can observe from \eqref{eqsysb}, the dynamic system that drives the economy over time identifies five variables, two state and three control variables. Also, as claimed by the above proposition, along the balanced growth path, both control variables $u$ and $v$ are constant, but the state variables $k$ and $h$, as well as the other control variable $c$, grow at a positive common rate $r_*$. Therefore, we need to transform the three variables having a positive common growth rate along the balanced growth path into two new stationary variables, that is which remain constant along the balanced growth path. This approach is necessary in order to have the possibility to decide on the stability of the balanced growth path. The new stationary variables that can be defined are $z=\frac{k}{h}$ and $q = c/k$ and, together with the other two variables $u$ and $v$, enable us to formulate a new system of differential equations describing the dynamics of the economy around the steady-state.

From the first two equations of the system \eqref{eqsysb} we obtain
$$\frac{\dot{z}}{z} = \frac{\dot{k}}{k}-\frac{\dot{h}}{h} = A_1v\left(\theta w\right)^{-1}P_1^{\frac{1}{\psi_1}}-A_2P_2^{\frac{1}{\psi_2}}(1-u)-\delta_{k} +\delta_{h} - \frac{c}{k}=-D- \frac{c}{k}.$$
From the first and the third equations of the system \eqref{eqsysb} we get
$$\frac{\dot{q}}{q} = \frac{\dot{c}}{c}-\frac{\dot{k}}{k} = q-\frac{A_1 P_{\varepsilon}P_1^{\frac{1}{\psi_1}-1}}{\varepsilon\theta w} -\frac{\rho-(\varepsilon-1)\delta_{k}}{\varepsilon},$$
and thus we obtain the following differential system of stationary variables.
\begin{equation}\label{eqsysbgp}
\left\{
  \begin{array}{llll}
\dot{z}= -\left\{D + q\right\}z\\\\
\dot{q} = \left\{q-\frac{A_1 P_{\varepsilon}P_1^{\frac{1}{\psi_1}-1}}{\varepsilon\theta w} -\frac{\rho-(\varepsilon-1)\delta_{k}}{\varepsilon}\right\}q,\\\\
\dot{u}=\left[D+q+\frac{G_2QP}{R}\right]\frac{u(1-u)}{u-v},\\\\
\dot{v}=\left[D+q+\frac{G_1 QP}{R}\right]\frac{v(1-v)}{u-v}.\\
\end{array}
  \right.
\end{equation}
As we can observe from the equation \eqref{eqsr}, the new stationary variable $z$ depends explicitly on the control variables $u$ and $v$, and therefore the determinant of the Jacobian matrix, evaluated at steady state will be equal to zero. Consequently, the reduced system - from five to four variables, will not allow to claim the saddle-path stability. As it is well-known, a system is not necessarily saddle path stable if the determinant of the Jacobian is zero and a zero determinant doesn't give enough information to determine the stability type, and further analysis of the eigenvalues is needed. Also, a zero determinant can indicate a bifurcation point or other complex behavior.
\subsection{Some final conclusions and remarks}
If we consider that the two sectors are modeled by two distinct Cobb-Douglass production functions, that is
$$f= A_1\left(kv\right)^\beta\left(hu\right)^{1-\beta}\;\mbox{and}\; g=A_2\left[k(1-v)\right]^\alpha\left[h(1-u)\right]^{1-\alpha}$$ then we immediately obtain
\begin{equation}\label{eqvfu}
v=\frac{u}{u+\theta(1-u)},\;\theta = \frac{\alpha(1-\beta)}{\beta(1-\alpha)},
\end{equation}
and the four differential equations that describe the dynamical system that drives the economy over time
\begin{equation}\label{eqsysf}
\left\{
\begin{array}{lllll}
\frac{\dot{k}}{k} = \frac{A_1uw^{\beta-1}}{u+\theta(1-u)}-\frac{c}{k}-\delta_k,\\\\
\frac{\dot{h}}{h} = A_2\theta^\alpha(1-u)w^\alpha-\delta_h ,\\\\
\frac{\dot{c}}{c}=\frac{1}{\varepsilon}\left(A_1\beta w^{\beta-1}-\delta_k-\rho\right),\\\\
\frac{\dot{u}}{u}=\left[(\beta-\alpha)\left(D-\frac{c}{k}\right)+
P\right]\frac{u+\theta(1-u)}{(\beta-\alpha)(1-\theta)u},\\
\end{array}
\right.
\end{equation}
where
$$D=\frac{A_1uw^{\beta-1}}{u+\theta(1-u)}-A_2\theta^\alpha (1-u)w^\alpha+\delta_h-\delta_k,$$
$$P=A_2(1-\alpha)\theta^\alpha w^\alpha-A_1\beta w^{\beta-1}+\delta_k-\delta_h,$$
and the two differential equations describing the trajectories of the dual variables.
$$\frac{\dot{\lambda}}{\lambda}=\rho+\delta_{k}-A_1\beta w^{\beta-1},\\\\
\frac{\dot{\mu}}{\mu}= \rho+\delta_{h}-A_2(1-\alpha)\theta^\alpha w^\alpha.$$
As in the case of the first model, this system reaches the unique balanced growth path and the following statements are valid:
\begin{enumerate}
\item[i.] $r_{k^*} = r_{h^*} = r_{c^*}  = r_{y^*_1} = r_{y^*_2} = r_*$, where $y_1$ and $y_2$ represent the two $CD$ production functions and
\begin{equation}\label{eqsolrcd}
r_*=\frac{1}{\varepsilon}\left[A_1\beta w^{\beta-1}-\delta_k-\rho\right]
\end{equation}
\item [ii.] $w_{*}$ is the unique solution of the equation $P(w) = 0$, where
\begin{equation}\label{eqsolwcd}
P(w)=A_2(1-\alpha)\theta^\alpha w^\alpha-A_1\beta w^{\beta-1}+\delta_k-\delta_h,
\end{equation}
\item[iii.] $u_* \in [0, 1]$, $v_* \in [0, 1]$ and
\begin{equation}\label{eqsolucd}
u_* = 1-\frac{r_*+\delta_h}{A_2\theta^{\alpha}w_*^\alpha},\;v_* = \frac{u_*}{u_*+\theta\left(1-u_*\right)},
\end{equation}
\item[iv.]
\begin{equation}\label{eqsolc/kcd}
q_*=\frac{c_*}{k_*} = \frac{1}{\varepsilon} \left[A_1w^{\beta-1}P_{\varepsilon} +\rho-\delta_k\left(\varepsilon-1\right)\right],
\end{equation}
where $$P_{\varepsilon}=\frac{(\varepsilon-\beta)u_*-\beta\theta(1-u_*)}{u_*+\theta(1-u_*)}.$$
\end{enumerate}
(The proof is similar and therefore is not presented here.)
The dynamics of the economy around the steady-state, in terms of the stationary variables $z=\frac{k}{h}$, $q = c/k$ and $u$ - which is constant at the balanced growth path - is described by the following system of equations:
\begin{equation}\label{eqldrs}
\left\{
\begin{array}{lll}
\dot{z}=\left[D(z,u)-q\right]z,\\\\
\dot{q}=\left[-\frac{A_1 H(z,u)}{\varepsilon}+q+\frac{(\varepsilon-1)\delta_k-\rho}{\varepsilon}\right]q,\\\\
\dot{u}=\left\{(\beta-\alpha)\left[D(z,u)-q\right]+P(z,u)\right\}\frac{u+\theta(1-u)}{(\beta-\alpha)(1-\theta)},\\
\end{array}
\right.
\end{equation}
where
$$D(z,u)=\frac{A_1uz^{\beta-1}}{\left[u+\theta(1-u)\right]^\beta}-\frac{A_2\theta^\alpha (1-u)z^\alpha}{\left[u+\theta(1-u)\right]^\alpha}+\delta_h-\delta_k,$$
$$P(z,u)=\frac{A_2(1-\alpha)\theta^\alpha z^\alpha}{\left[u+\theta(1-u)\right]^\alpha}-\frac{A_1\beta z^{\beta-1}}{\left[u+\theta(1-u)\right]^{\beta-1}}+\delta_k-\delta_h,$$
$$H(z,u) = \frac{z^{\beta-1}P_\varepsilon}{\left[u+\theta(1-u)\right]^{\beta-1}}.$$
Unfortunately, the results obtained for the Jacobian matrix do not allow us for an incontestable interpretation of the signs of these eigenvalues. Nevertheless, all the numerical simulations we made, confirm that at least one eigenvalue is negative and therefore, we can claim that there exists a unique optimal steady-state equilibrium that is saddle-path stable.
\begin{rem}
If we choose the following two cases for the benchmark values for the economy:
\begin{enumerate}
  \item $A_1=1.05,\;A_2=0.20\;\alpha=0.45,\;\beta=0.75,\;\delta_k=0.06,\;\delta_h=0.05,\;\varepsilon=2,
      \;\rho=0.06.$
  \item $A_1=1.05,\;A_2=0.20\;\alpha=0.75,\;\beta=0.45,\;\delta_k=0.06,\;\delta_h=0.05,\;\varepsilon=2,
      \;\rho=0.06$
\end{enumerate}
and the results of the numerical simulations for the two cases, are:
\begin{itemize}
  \item[1] $z_* = 29.57, u_* = 0.73, v_*= 0.91, q_*= 0.23$,
  $EV=\left[0.16; 0.89; -.73\right]$,
  \item[2] $z_* = 3.42, u_* = 0.88, v_*= 0.67, q_*= 0.28$,
  $EV=\left[0.14; 1.19; -1.05\right]$,
\end{itemize}
where $EV$ signify eigenvalues.
\end{rem}
\begin{rem}
If we consider the case of the same elasticity of substitution, that is $\psi_1 = \psi_2 =\psi > 1$, but the other parameters are distinct, that is $A_1\neq A_2, \alpha_1 \neq \alpha_2$, then the results confirm the hypothesis of Bond et al.(see the recent paper of Chilarescu $2025$).
\end{rem}


\begin{thebibliography}{99}
\bibitem{Ben} Benhabib J. and Perli R., $1994$. Uniqueness and
Indeterminacy: On the Dynamics of Endogenous Growth. {\it Journal of
Economic Theory}, $63, 113 - 142$.
\bibitem{Bon} Bond, E. W., Wang, P. and Yip, C. K., 1996. A General
Two-Sector Model of Endogenous Growth with Human and Physical
Capital: Balanced Growth and Transitional Dynamics, {\it Journal
of Economic Theory} $68, 149 - 173$.
\bibitem{Chi} Chilarescu C., $2025$. Elasticity of substitution and economic growth: Some new results.
{\it Mathematical Methods in the Applied Sciences}, $48$ $(5)$, $5896 - 5905$.
\bibitem{Reb} Rebelo, S.,  1991. Long-Run Policy Analysis and
Long-Run Growth, {\it Journal of Political Economy} 99, 500 - 521.
\end{thebibliography}
\end{document}